\documentclass{amsart}
\usepackage{amsmath,amsfonts,amssymb,amsthm}
\usepackage{mathrsfs,geometry}
\usepackage[T1]{fontenc}
\usepackage[latin1]{inputenc}
\usepackage{url}
\usepackage{graphicx}

\newcommand{\iy}{\ensuremath{\infty}}
\newcommand{\R}{\ensuremath{\mathbf{R}}}

\newcommand{\N}{\ensuremath{\mathbf{N}}}

\renewcommand{\P}{\ensuremath{\mathbf{P}}}
\newcommand{\e}{\varepsilon}
\newcommand{\st}{\textnormal{ s.t. }}
\renewcommand{\leq}{\leqslant}
\renewcommand{\geq}{\geqslant}

\newcommand{\stochprec}{\leq_\textnormal{st}}
\newcommand{\notstochprec}{\nleqslant_\textnormal{st}}
\newcommand{\Mstochprec}{\leq^*_{\textnormal{st}}}
\newcommand{\Cstochprec}{\leq^{\textnormal{C}}_{\textnormal{st}}}
\DeclareMathOperator{\supp}{supp}
 
\DeclareMathOperator{\Prob}{P}
\newcommand{\el}{e_{\lambda}}

\newcommand{\PE}{\textnormal{P}_{\exp}(\R)}

\newtheorem{theo}{Theorem}[section]
\newtheorem{defn}[theo]{Definition}
\newtheorem{pr}[theo]{Proposition}
\newtheorem{lemma}[theo]{Lemma}
\newtheorem{rk}[theo]{Remark}
\newtheorem{cor}[theo]{Corollary}

\newtheorem{conj}[theo]{Conjecture}

\newtheorem{ex}[theo]{Example}
\newtheorem*{theo*}{Theorem}

\author{Guillaume AUBRUN and Ion NECHITA}
\thanks{Research was supported in part by the European Network Phenomena in High Dimensions, FP6
Marie Curie Actions, MCRN-511953}
\title{Stochastic domination for iterated convolutions and catalytic majorization}

\begin{document}

\begin{abstract}
We study how iterated convolutions of probability measures compare under stochastic domination. We give necessary and sufficient conditions for the existence of an integer $n$ such that $\mu^{*n}$ is stochastically dominated by $\nu^{*n}$ for two given probability measures $\mu$ and $\nu$. As a consequence we obtain a similar theorem on the majorization order for vectors in $\R^d$. In particular we prove results about catalysis in quantum information theory.
\medskip

\textbf{Domination stochastique pour les convolutions itérées et catalyse quantique}

\noindent
\textsc{Résumé.} Nous étudions comment les convolutions itérées des mesures de probabilités se comparent pour la domination stochastique. Nous donnons des conditions nécessaires et suffisantes pour l'existence d'un entier $n$ tel que $\mu^{*n}$ soit stochastiquement dominée par $\nu^{*n}$, étant données deux mesures de probabilités $\mu$ et $\nu$. Nous obtenons en corollaire un théorème similaire pour des vecteurs de $\R^d$ et la relation de Schur-domination. Plus spécifiquement, nous démontrons des résultats sur la catalyse en théorie quantique de l'information. 
\end{abstract}
\maketitle

\section*{Introduction and notations}
This work is a continuation of \cite{an}, where we study the phenomenon of catalytic majorization in quantum information theory. A probabilistic approach to this question involves stochastic domination which we introduce in Section \ref{sec:stoch_dom} and its behavior with respect to the convolution of measures. We give in Section \ref{sec:results} a condition on measures $\mu$ and $\nu$ for the existence of an integer $n$ such that $\mu^{*n}$ is stochastically dominated by $\nu^{*n}$. We gather further topological and geometrical aspects in Section \ref{sec:top_geom}. Finally, we apply these results to our original problem of catalytic majorization. In Section \ref{sec:quantum} we introduce the background for quantum catalytic majorization and we state our results. Section \ref{sec:proofs} contains the proofs and in Section \ref{sec:inf_cat} we consider an infinite dimensional version of catalysis.

We introduce now some notation and recall basic facts about probability measures. We write $\Prob(\R)$ for the set of probability measures on $\R$. We denote by $\delta_x$ the Dirac mass at point $x$. If $\mu \in \Prob(\R)$, we write $\supp \mu$ for the support of $\mu$. We write respectively $\min \mu \in [-\iy,+\iy)$ and $\max \mu \in (-\iy,+\iy]$ for $\min \supp \mu$ and $\max \supp \mu$. We also write $\mu(a,b)$ and $\mu[a,b]$ as a shortcut for $\mu((a,b))$ and $\mu([a,b])$. 
The convolution of two measures $\mu$ and $\nu$ is denoted $\mu * \nu$. Recall that if $X$ and $Y$ are independent random variables of respective laws $\mu$ and $\nu$, the law of $X+Y$ is given by $\mu * \nu$. The results of this paper are stated for convolutions of measures, they admit immediate translations in the language of sums of independent random variables. For $\lambda \in \R$, the function $\el$ is defined by $\el(x)=\exp(\lambda x)$. 

\section{Stochastic domination}\label{sec:stoch_dom}

A natural way of comparing two probability measures is given by the following relation
\begin{defn}
Let $\mu$ and $\nu$ be two probability measures on the real line. We say that $\mu$ is \emph{stochastically dominated} by $\nu$ and we write $\mu \stochprec \nu$ if
\begin{equation} \label{def-stochasticordering} \forall t \in \R, \mu[t, \iy) \leq \nu[t, \iy).\end{equation}
\end{defn}

Stochastic domination is an order relation on $\Prob(\R)$ (in particular, $\mu \stochprec \nu$ and $\nu \stochprec \mu$ imply $\mu=\nu$). The following result \cite{stoyan,gs} provides useful characterizations of stochastic domination.

\begin{theo*}
Let $\mu$ and $\nu$ be probability measures on the real line. The following are equivalent
\begin{enumerate}
 \item $\mu \stochprec \nu$.
 \item {\bf Sample path characterization.} There exists a probability space $(\Omega,\mathcal F,\P)$ and two random variables $X$ and $Y$ on $\Omega$ with respective
laws $\mu$ and $\nu$, so that
\[ \forall \omega \in \Omega, X(\omega) \leq Y(\omega) .\]
 \item {\bf Functional characterization.} For any increasing function $f : \R \to \R$ so that both integrals exist,
\[ \int f d\mu \leq \int f d\nu .\]
\end{enumerate}
 \end{theo*}
It is easily checked that stochastic domination is well-behaved with respect to convolution.

\begin{lemma}
\label{convolution}
Let $\mu_1$, $\mu_2$, $\nu_1$, $\nu_2$ be probability measures on the real line. If $\mu_1 \stochprec \nu_1$ and $\mu_2 \stochprec \nu_2$, then $\mu_1 * \mu_2 \stochprec \nu_1 * \nu_2$.
\end{lemma}

\begin{lemma}
\label{conv-power}
Let $\mu$ and $\nu$ be two probability measures on the real line such that $\mu \stochprec \nu$. Then, for all $n \geq 2$, $\mu^{*n} \stochprec \nu^{*n}$.
\end{lemma}

For fixed $\mu$ and $\nu$, it follows from Lemma \ref{convolution} that the set of integers $k$ so that $\mu^{*k} \stochprec \nu^{*k}$ is stable under addition. In general $\mu^{*n} \stochprec \nu^{*n}$ does not imply $\mu^{*(n+1)} \stochprec \nu^{*(n+1)}$. Here is a typical example.

\begin{ex}\label{ex:1239}
Let $\mu$ and $\nu$ be the probability measures defined as
\[ \mu = 0.4 \delta_0 + 0.6 \delta_2 \]
\[ \nu = 0.8 \delta_1 + 0.2 \delta_3 \]
It is straightforward to verify (see Figure \ref{fig:cum_dist_1239}) that
\begin{itemize}
 \item For $k=2$, and therefore for all even $k$, we have $\mu^{*k} \stochprec \nu^{*k}$.
 \item For $k$ odd, we have $\mu^{*k} \stochprec \nu^{*k}$ only for $k \geq 9$.
\end{itemize}
\begin{figure}[h]
\includegraphics[width=3.5cm,height=4cm]{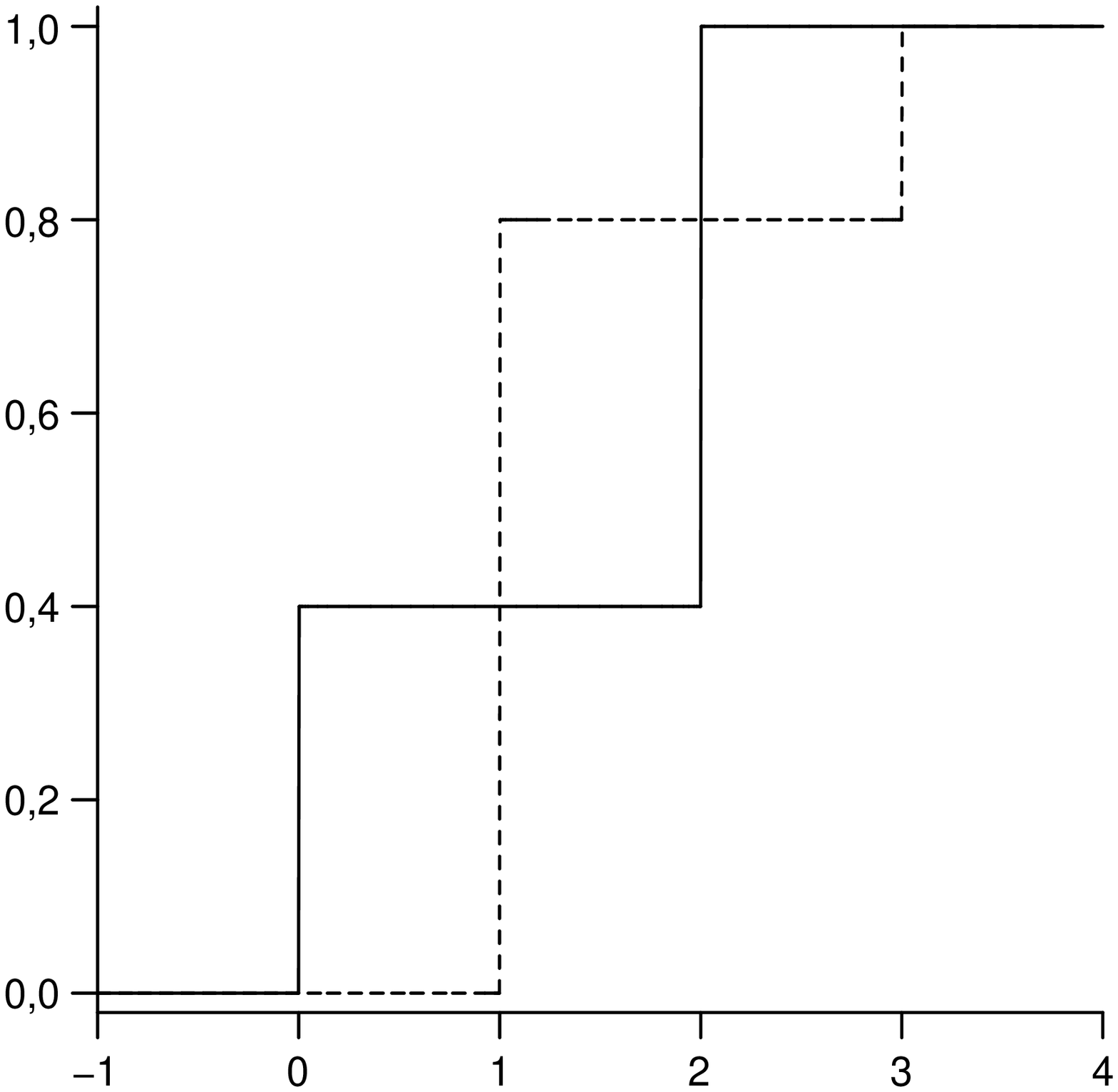}
\includegraphics[width=3.5cm,height=4cm]{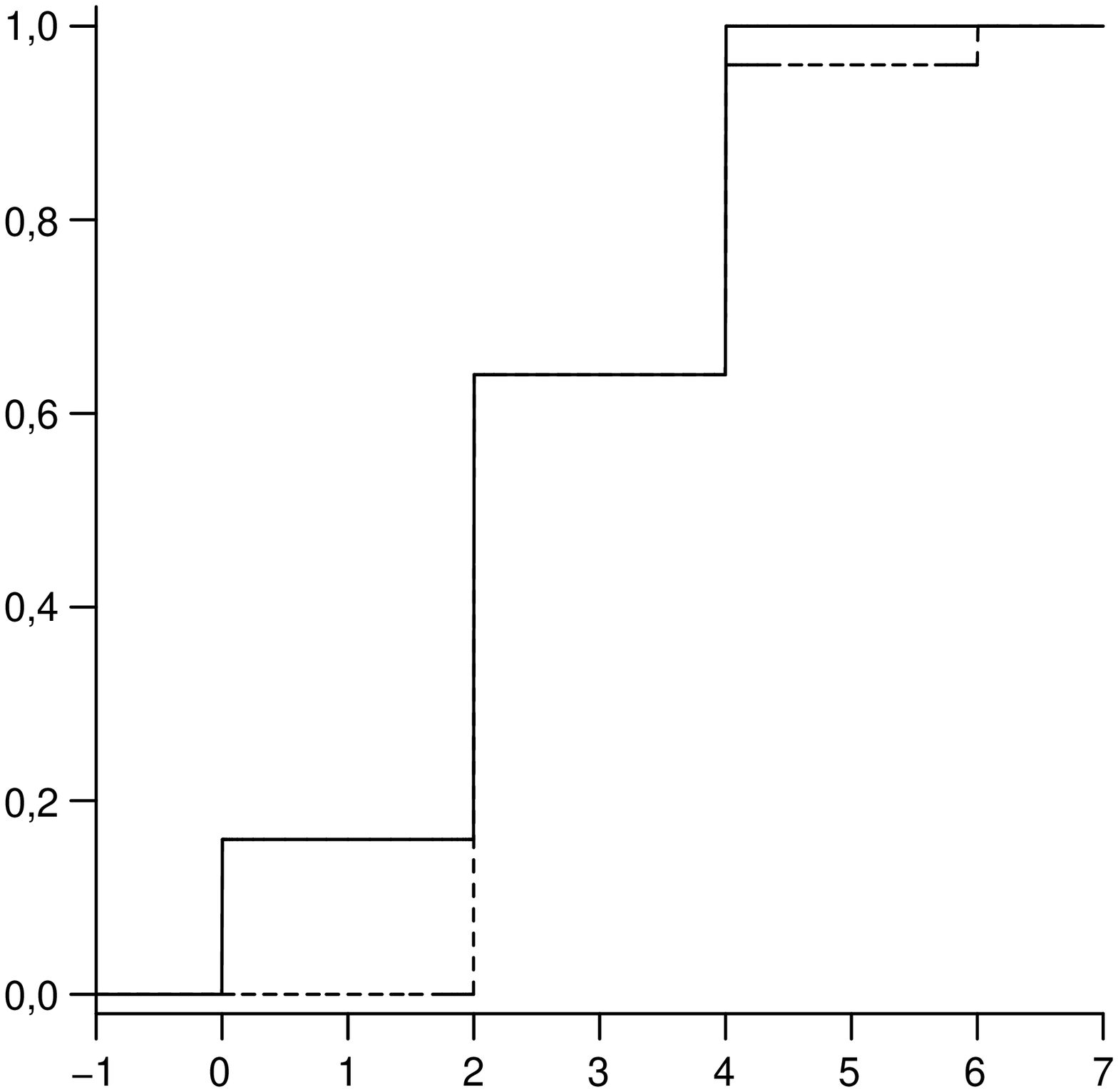}
\includegraphics[width=3.5cm,height=4cm]{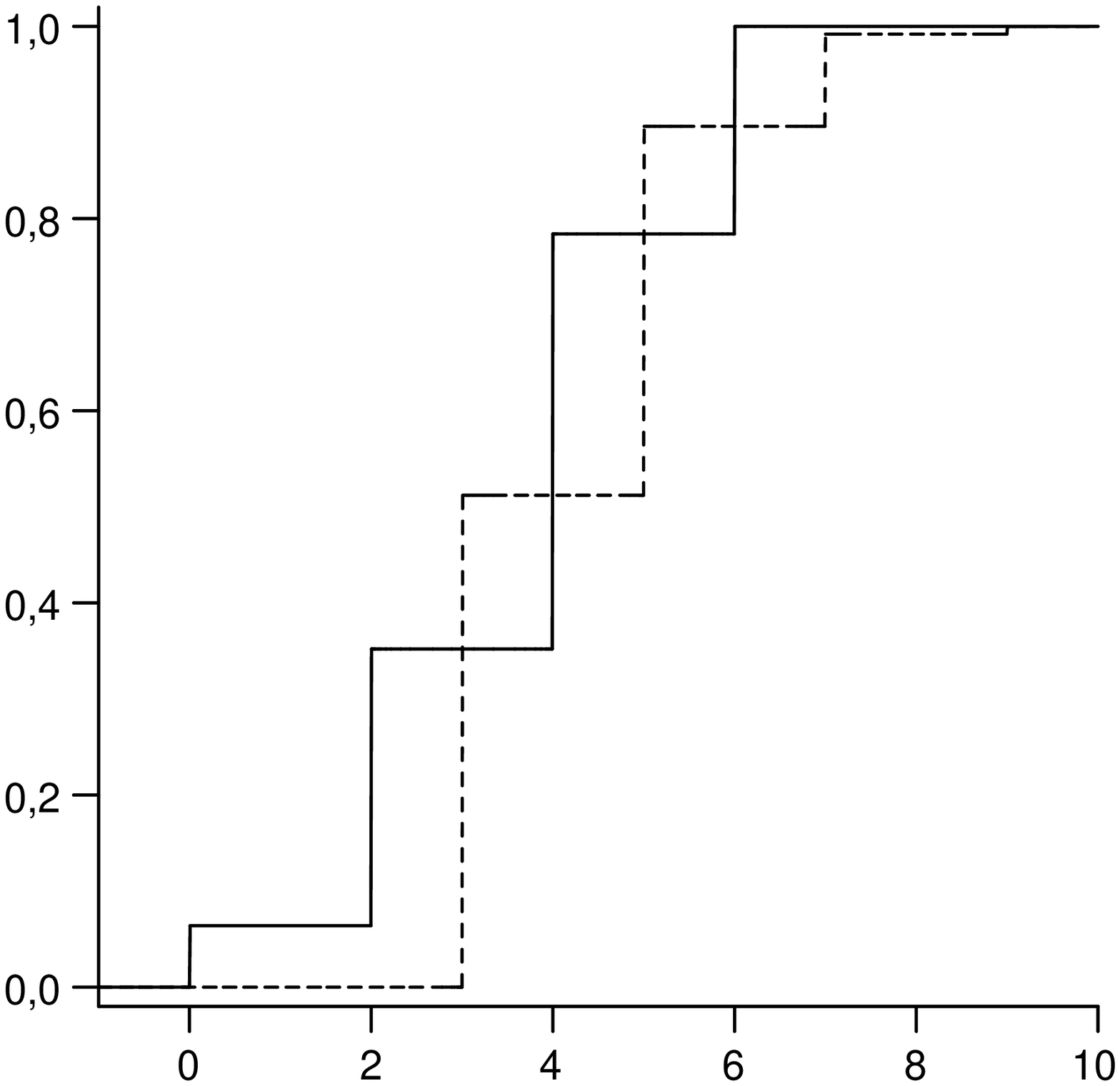}
\includegraphics[width=3.5cm,height=4cm]{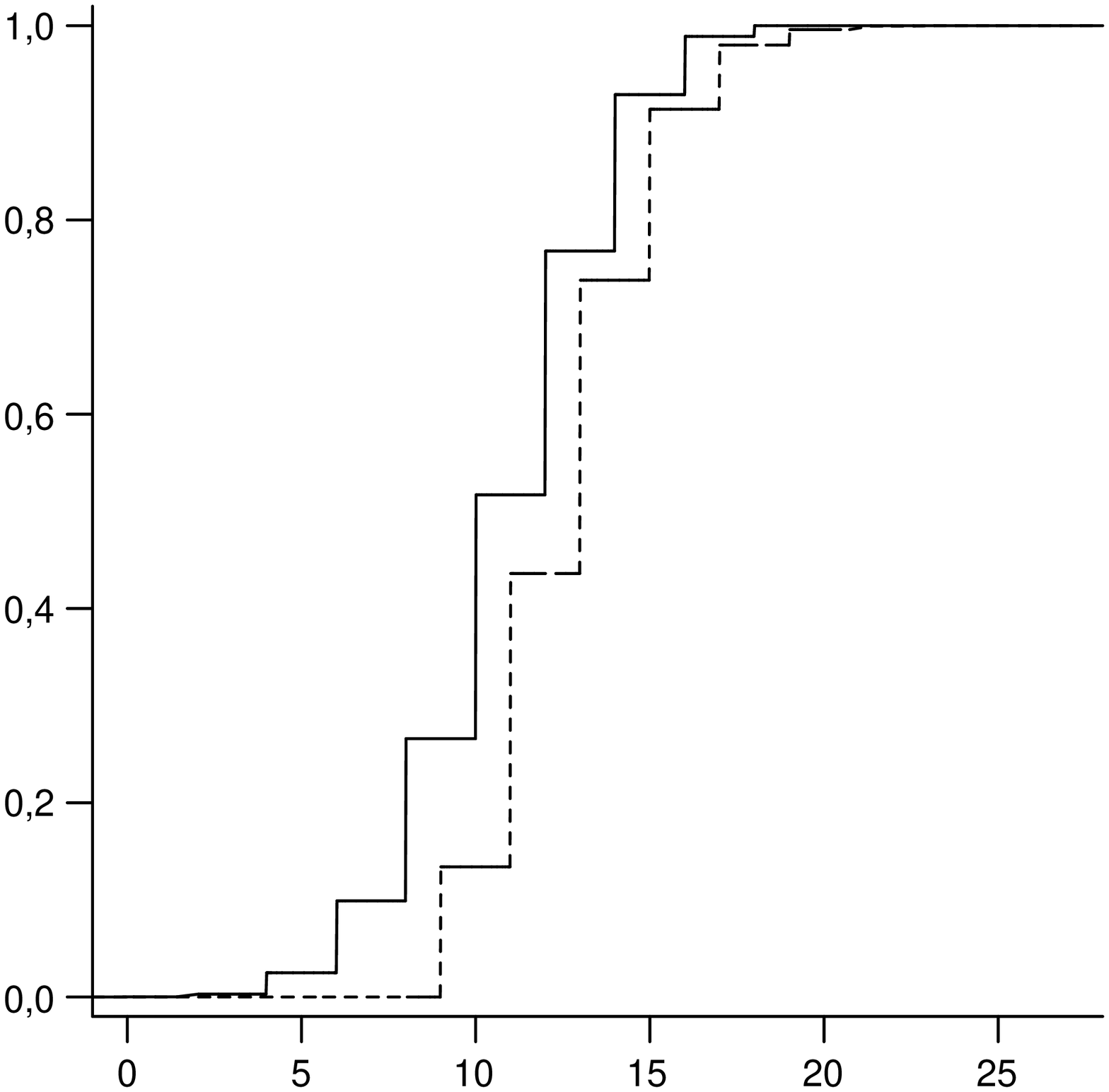}
\caption{Cumulative distribution functions of of $\mu^{*k}$ (solid line) and $\nu^{*k}$ (dotted line) from Example \ref{ex:1239} for $k=1,2,3,9$.}\label{fig:cum_dist_1239}
\end{figure}
\end{ex}

Other examples show that the minimal $n$ so that $\mu^{*n} \stochprec \nu^{*n}$ can be arbitrary large. This is the content of the next proposition.

\begin{pr}
\label{2points}
For every integer $n$, there exist compactly supported probability measures $\mu$ and $\nu$ such that 
$\mu^{*n} \stochprec \nu^{*n}$ and, for all $1 \leq k \leq n-1$, $\mu^{*k} \notstochprec \nu^{*k}$. 
\end{pr}

\begin{proof}
Let $\mu = \e \delta_{-2n} + (1-\e)\delta_1$ and $\nu$ be the uniform measure on $[0,2]$, where $0<\e<1$ will be defined later. For $k \geq 1$, 
\[\mu^{*k} = \sum_{i=0}^{k} \binom{k}{i} (1-\e)^i\e^{k-i}\delta_{i-2n(k-i)},\]
Note that $\supp(\nu^{*k}) \subset \R^+$, while for $1 \leq k \leq n$, the only part of $\mu^{*k}$ charging $\R_+$ is the Dirac mass at point $k$. This implies that 
\[ \mu^{*k} \stochprec \nu^{*k} \iff \mu^{*k}[k,+\iy) \leq \nu^{*k}[k,+\iy) .\]
We have $\mu^{*k}[k,+\iy) = (1-\e)^k$ and $\nu^{*k}[k,+\iy)=1/2$. It remains to choose $\e$ so that $(1-\e)^n < 1/2 <(1-\e)^{n-1}$.
\end{proof}


\section{Stochastic domination for iterated convolutions and Cramér's theorem}\label{sec:results}

In light of previous examples, we are going to study the following extension of stochastic domination

\begin{defn}
 We define a relation $\Mstochprec$ on $\Prob(\R)$ as follows
\[ \mu \Mstochprec \nu \iff \exists n \geq 1 \st  \mu^{*n} \stochprec \nu^{*n} .\]
\end{defn}

In turns that when defined on $\Prob(\R)$, this relation is not an order relation due to pathological poorly integrable measures. Indeed, there exist  two probability measures $\mu$ and $\nu$ so that $\mu \neq \nu$ and $\mu * \mu = \nu * \nu$ (see \cite{feller}, p. 479). Therefore, the relation $\Mstochprec$ is not anti-symmetric. For this reason, we restrict ourselves to sufficiently integrable measures (however, most of what follows generalizes to wider classes of measures). This is quite usual when studying orderings of probability measures, see \cite{stoyan} for examples of such situations.

\begin{defn}
A  measure $\mu$ on $\R$ is said to be \emph{exponentially integrable} if $\int \el d\mu < +\iy$ for all $\lambda \in \R$ (recall that $\el(x)=\exp(\lambda x)$). We write $\PE$ for the set of exponentially integrable probability measures.
\end{defn}

Notice that the space of exponentially integrable measures is stable under convolution. 

\begin{pr}
When restricted to $\PE$, the relation $\Mstochprec$ is a partial order.
\end{pr}
\begin{proof}
One has to check only the antisymmetry property, the other two being obvious. Let $k$ and $l$ be two integers such that $\mu^{*k} \stochprec \nu^{*k}$ and $\nu^{*l} \stochprec \mu^{*l}$. Then $\mu^{*kl} \stochprec \nu^{*kl} \stochprec \mu^{*kl}$ and therefore $\mu^{*kl} = \nu^{*kl}$. But if $\mu$ and $\nu$ are exponentially integrable, this implies that $\mu=\nu$. One can see this in the following way: if we denote the moments of $\mu$ by $m_p(\mu) = \int x^p d\mu(x)$, one checks by induction on $p$ that $m_p(\mu)=m_p(\nu)$ for all $p \in \N$. On the other hand, exponential integrability implies that $m_{2p}(\mu)^{1/2p} \leq C p$ for some constant $C$, so that Carleman's condition is satisfied (see \cite{feller}, p. 224). Therefore $\mu$ is determined by its moments and $\mu=\nu$.
\end{proof}

We would like to give a description of the relation $\Mstochprec$, for example similar to the functional characterization of $\stochprec$. We start with the following lemma

\begin{lemma}
\label{lemma-easydirection}
Let $\mu,\nu \in \PE$  such that $\mu \Mstochprec \nu$. Then the following inequalities hold:
\begin{enumerate}
 \item[(a)] $\forall \lambda > 0, \int \el d\mu \leq \int \el d\nu$,
 \item[(b)] $\forall \lambda < 0, \int \el d\mu \geq \int \el d\nu$,
 \item[(c)] $\int x d\mu(x) \leq \int x d\nu(x)$,
 \item[(d)] $\min \mu \leq \min \nu$,
 \item[(e)] $\max \mu \leq \max \nu$,
\end{enumerate}
\end{lemma}

\begin{proof}
Let $\mu \Mstochprec \nu$ and $\lambda > 0$. Since $\mu^{*n} \leq \nu^{*n}$ for some $n$, we get from the functional characterization of $\stochprec$ that
\[ \int \el d\mu^{*n} \leq \int \el d\nu^{*n}. \]
It remains to notice that
\[ \int \el d\mu^{*n} = \left( \int \el d\mu \right)^n \]
and we get (a). The proof of (b) is completely symmetric, while (c) follows also from the functional characterization. Conditions (d) and (e) are obvious since $\min(\mu^{*n}) = n \min (\mu)$ and $\max(\mu^{*n}) = n \max (\mu)$.
\end{proof}

The following Proposition shows that the necessary conditions of Lemma \ref{lemma-easydirection} are ``almost sufficient''.

\begin{pr}
\label{probabilisticlemma}
Let $\mu,\nu \in \PE$. Assume that the following inequalities hold 
\begin{enumerate}
\item[(a)] $\forall \lambda > 0, \int \el d\mu < \int \el d\nu $.
\item[(b)] $\forall \lambda < 0, \int \el d\nu < \int \el d\mu $.
\item[(c)] $\int x d\mu(x) < \int x d\nu(x)$.
\item[(d)] $\max \mu < \max \nu$.
\item[(e)] $\min \mu < \min \nu$.
\end{enumerate}
Then $\mu \Mstochprec \nu$, and more precisely there exists an integer $N \in \N$ such that for any $n\geq N$, $\mu^{*n} \stochprec \nu^{*n}$.
\end{pr}

We give in Proposition \ref{pr-notclosed} a counter-example showing that Proposition \ref{probabilisticlemma} is not true when stated with large inequalities.

\medskip

We are going to use Cramér's theorem on large deviations. The cumulant generating function 
$\Lambda_\mu$ of the probability measure $\mu$ is defined for any $\lambda \in 
\R$ by
\[ \Lambda_\mu(\lambda) = \log \int \el d\mu .\]
It is a convex function taking values in $\R$. Its convex conjugate 
$\Lambda_\mu^*$, sometimes called the Cramér transform,  is defined as
\begin{equation*} 
 \Lambda_\mu^*(t) = \sup_{\lambda \in \R} \lambda t - \Lambda_\mu(\lambda). 
\end{equation*}
Note that $\Lambda_\mu^*: \R \to [0, +\iy]$ is a smooth convex function, which takes the value $+\iy$ on $\R \setminus [\min \mu,\max \mu]$. 
Moreover, for $t \in (\min \mu,\max \mu)$, the supremum in the definition of 
$\Lambda_\mu^*(t)$ is attained at a unique point $\lambda_t$. Moreover, 
$\lambda_t > 0$ if $t > \int x d\mu(x)$ and $\lambda_t <0$ if $t < \int x d\mu(x)$. Also, $\Lambda_\mu^*(\int x d\mu(x)) = 0$ since $\Lambda_\mu'(0) = \int x d\mu(x)$.
We now state Cramér's theorem. The theorem can be equivalently stated in the language of sums of i.i.d. random variables \cite{dz,gs}.

\begin{theo*}[Cramér's theorem]
Let $\mu \in \PE$. Then for any $t \in \R$, 
\begin{equation} \lim_{n \to \iy} \frac{1}{n} \log \mu^{*n}[tn,+\iy)
 = \begin{cases} 0 & \textnormal{if } t \leq \int x d\mu(x)  \\ 
-\Lambda_X^*(t) & \textnormal{otherwise. } \end{cases} 
\label{cramer-uppertail} \end{equation}
\begin{equation} \lim_{n \to \iy} \frac{1}{n} \log \left(1-\mu^{*n} (tn,+\iy)
\right) = \begin{cases} 0 & \textnormal{if } t \geq \int x d\mu(x)  \\ 
-\Lambda_X^*(t) & \textnormal{otherwise. } \end{cases} 
\label{cramer-lowertail} \end{equation}
\end{theo*}

\begin{proof}[Proof of Proposition \ref{probabilisticlemma}]
Note that the hypotheses imply that the quantities $\max \mu$ and $\min \nu$ are finite.
We write also $M_\mu = \int x d\mu(x)$ and $M_\nu = \int x d\nu(x)$.
For $n\geq 1$, define $(f_n)$ and $(g_n)$ by
\[ f_n(t) = \mu^{*n} [tn,+\iy) , \]
\[ g_n(t) = \nu^{*n} [tn,+\iy) . \]
We need to prove that $f_n \leq g_n$ on $\R$ for $n$ large enough. If 
$t > \max \mu$, the inequality is trivial since $f_n(t) = 0$. Similarly, 
if $t < \min \nu$ we have $g_n(t) =1$ and there is nothing to prove.

Fix a real number $t_0$ such that $M_\mu < t_0 < M_\nu$. We first work on 
the interval $I = [t_0, \max \mu]$. By Cramér's theorem, the sequences $(f_n^{1/n})$ and 
$(g_n^{1/n})$ converge respectively on $I$ toward $f$ and $g$ defined 
by
\[ f(t) = \exp (-\Lambda_\mu^*(t)) ,\] 
\[ g(t) = \begin{cases} 1 & \textnormal{if } t_0 \leq t \leq M_\nu \\ 
\exp(-\Lambda_\nu^*(t)) & \textnormal{if }  M_\nu \leq t \leq \max \mu .
\end{cases} \]
Note that $f$ and $g$ are continuous on $I$. We claim also that $f < g$ 
on $I$. The inequality is clear on $[t_0,M_\nu]$ since $f<1$. If $t \in 
(M_\nu,\max \mu]$, note that the supremum in the definition of 
$\Lambda_\nu^*(t)$ is attained for some $\lambda >0$ --- to show this we 
used hypothesis (d). Using (a) and the definition of the convex 
conjugate, it implies that $\Lambda^*_\nu(t) > \Lambda^*_\mu(t)$. 
We now use the following elementary fact: if a sequence of 
non-increasing functions defined on a compact interval $I$ converges 
pointwise toward a continuous limit, then the convergence is actually 
uniform on $I$ (for a proof see \cite{ps} Part 2, Problem 127; this 
statement is attributed to P\'olya or to Dini depending on authors). We 
apply this result to both $(f_n^{1/n})$ and $(g_n^{1/n})$ ; and 
since $f<g$, uniform 
convergence implies that for $n$ large enough, $f_n^{1/n} < g_n^{1/n}$ 
on $I$, and thus $f_n \leq g_n$.

Finally, we apply a similar argument on the interval $J = [\min 
\nu,t_0]$, except that we consider the sequences $(1-f_n)^{1/n}$ and 
$(1-g_n)^{1/n}$, and we use \eqref{cramer-lowertail} to compute the 
limit. We omit the details since the argument is totally symmetric.

We eventually showed that for $n$ large enough, $f_n \leq g_n$ on $I 
\cup J$, and thus on $\R$. This is exactly the conclusion of the proposition.
\end{proof}

\section{Geometry and topology of $\Mstochprec$}\label{sec:top_geom}

We investigate here the topology of the relation $\Mstochprec$. We first need to define a adequate topology on $\PE$. This space can be topologized in several ways, an important point for us being that the map $\mu \mapsto \int \el d\mu$ should be continuous.

\begin{defn}
A function $f : \R \to \R$ is said to be subexponential if there exist constants $c,C$ so that for every $x \in \R$
\[ |f(x)| \leq C \exp (c|x|). \]
\end{defn}

\begin{defn}
Let $\tau$ be the topology defined on the space of exponentially integrable measures, generated by the family of seminorms $(N_f)$
\[ N_f(\mu) = \left|\int f d\mu\right| ,\]
where $f$ belongs to the class of continuous subexponential functions.
\end{defn}

The topology $\tau$ is a locally convex vector space topology. It can be shown that the relation $\Mstochprec$ is not $\tau$-closed (see Proposition \ref{pr-notclosed}). However, we can give a functional characterization of its closure. This is the content of the following theorem.

\begin{theo} \label{theo-measures}
Let $R \subset \PE^2$ be the set of couples $(\mu,\nu)$ of exponentially integrable probability measures so that $\mu \Mstochprec \nu$. Then
\begin{equation} \label{closure}
\overline{R} = \left\{ (\mu,\nu) \in \PE^2 \st \forall \lambda \geq 0, \int \el d\mu \leq \int \el d\nu  \text{ and } \forall \lambda \leq 0, \int \el d\mu \geq \int \el d\nu \right\} ,\end{equation}
the closure being taken with respect to the topology $\tau$.
\end{theo}

\begin{proof}
Let us write $X$ for the set on the right-hand side of \eqref{closure}. We get from Lemma \ref{lemma-easydirection} that $R \subset X$. Moreover, it is easily checked that $X$ is $\tau$-closed, therefore $\overline{R} \subset X$. Conversely, we are going to show that the set of couples $(\mu,\nu)$ satisfying the hypotheses of Proposition \ref{probabilisticlemma} is $\tau$-dense in $X$. Let $(\mu,\nu) \in X$. We get from the inequalities satisfied by $\mu$ and $\nu$ that
\begin{itemize}
 \item $\int x d\mu(x) \leq xd\nu (x)$ (taking derivatives at $\lambda=0$),
 \item $\min \mu \leq \min \nu$ (taking $\lambda \to -\iy$),
 \item $\max \mu \leq \max \nu$ (taking $\lambda \to +\iy$).
\end{itemize}
We want to define two sequences $(\mu_n,\nu_n)$ which $\tau$-converge toward $(\mu,\nu)$, with $\mu_n \stochprec \mu$ and $\nu \stochprec \nu_n$ and for which the above inequalities become strict. Assume for example that $\max \mu = \max \nu = +\iy$ and $\min \mu = \min \nu = -\iy$. Then we can define $\mu_n$ and $\nu_n$ as follows: let $\e_n = \mu[n,+\iy)$ and $\eta_n=\nu(-\iy,-n]$, and set
\[ \mu_n = \mu_{|(-\iy,n)} + \e_n \delta_{n} ,\]
\[ \nu_n = \nu_{|(-n,+\iy)} + \eta_n \delta_{-n}. \]
We check using dominated convergence than $\lim \mu_n = \mu$ and $\lim \nu_n=\nu$ with respect to $\tau$, while by Proposition \ref{probabilisticlemma} we have $\mu_n \Mstochprec \nu_n$. The other cases are treated in a similar way: we can always play with small Dirac masses to make all inequalities strict (for example, if $\max \mu=\max \nu = M < +\iy$, replace $\nu$ by $(1-\e)\nu + \e \delta_{M+1}$, and so on). 
\end{proof}

A more comfortable way of describing the relation $\Mstochprec$ is given by the following sets

\begin{defn}
 Let $\nu \in \PE$. We define $D(\nu)$ to be the following set
\[ D(\nu) = \{ \mu \in \PE \st \mu \Mstochprec \nu \}. \]
\end{defn}

Using the ideas in the proof of Theorem \ref{theo-measures}, it can easily be showed that for $\nu \in \PE$ such that $\min \nu > - \iy$, one has 
\begin{equation}\label{eq:D_nu}
\overline{D(\nu)} = \left\{ \mu \in \PE \st \forall \lambda \geq 0, \int \el d\mu \leq \int \el d\nu  \text{ and } \forall \lambda \leq 0, \int \el d\mu \geq \int \el d\nu \right\},
\end{equation}
where the closure is taken in the topology $\tau$. However, for measures $\nu$ with $\min \nu = -\iy$, the condition (e) of Proposition \ref{probabilisticlemma} is violated and we do not know if the relation \eqref{eq:D_nu} holds. 

Another consequence of equation \eqref{eq:D_nu} is that the $\tau$-closure of $D(\nu)$ is a convex set. It is not clear that the set $D(\nu)$ itself is convex. We shall see in Proposition \ref{pr-notconvex} that this is not the case in general for measures $\nu \notin \PE$. Not also that for fixed $\nu \in \Prob(\R)$ the set $\{ \mu \in \Prob(\R) \st \mu \stochprec \nu \}$ is easily checked to be convex.

%

\begin{rk}
 One can analogously define for $\mu \in \PE$ the ``dual'' set 
\[ E(\mu) = \{ \nu \in \PE \st \mu \Mstochprec \nu \} .\]
Results about $D(\nu)$ or $E(\mu)$ are equivalent. Indeed, let $\mu^\leftrightarrow$ be the measure defined for a Borel set $B$ by $\mu^\leftrightarrow(B)=\mu(-B)$. We have $\mu \Mstochprec \nu \iff \nu^\leftrightarrow \Mstochprec \mu^\leftrightarrow$ and therefore $E(\mu)=D(\mu^\leftrightarrow)^\leftrightarrow$.
\end{rk}

We now give an example showing that the relation $\Mstochprec$ is not $\tau$-closed.

\begin{pr}
\label{pr-notclosed}
There exists a probability measure $\nu \in \PE$ so that the set $D(\nu)$ is not $\tau$-closed. Consequently, the set $R$ appearing in \eqref{closure} is not closed either.
\end{pr}

\begin{proof}
Let us start with a simplified sketch of the proof. By the examples of Section \ref{sec:stoch_dom}, for each positive integer $k$, one can find probability measures $\mu_k$ and $\nu_k$ such that $\mu_k \in D(\nu_k)$, while $\mu_k^{*k} \not \stochprec \nu_k^{*k}$. We sum properly rescaled and normalized versions of these measures in order to obtain two probability measures $\mu$ and $\nu$ such that $\mu \notin D(\nu)$. However, successive approximations $\tilde{\mu}_n$ of $\mu$ are shown to satisfy $\tilde{\mu}_n \stochprec \nu$ which implies $\mu \in \overline{D(\nu)}$ and thus $D(\nu) \neq \overline{D(\nu)}$.

We now work out the details. For $k \geq 1$, let $a_k = (k+2)!$, $b_k = (k+2)!+1$ and $\gamma_k = c \exp(-k^k)$, where the constant $c$ is chosen so that $\sum \gamma_k =1$.  We check that $(a_k)$ and $(b_k)$ satisfy the following inequalities
\begin{equation} \label{ineq1}
 (k-1) b_k + b_{k-1} < k a_k,
\end{equation}
\begin{equation} \label{ineq2}
 k b_k < a_{k+1}.
\end{equation}
It follows from Proposition \ref{2points} that for each $k \in \N$ there exist $\mu_k$ and $\nu_k$, probability measures with compact support such that $\mu_k \in D(\nu_k)$ while $\mu_k^{*k} \not \stochprec \nu_k^{*k}$. Moreover, we can assume that $\supp(\mu_k) \subset (a_k,b_k)$ and $\supp(\nu_k) \subset (a_k,b_k)$. Indeed, we can apply to both measures a suitable affine transformation (increasing affine transformations preserve stochastic domination and are compatible with convolution). We now define $\mu$ and $\nu$ as
\[ \mu = \sum_{k=1}^{\iy} \gamma_k \mu_k \ \ \ \ \textnormal{ and } \ \ \ \ \nu = \sum_{k=1}^{\iy} \gamma_k \nu_k .\]
Note that the sequence $(\gamma_k)$ has been chosen to tend very quickly to 0 to ensure that $\mu$ and $\nu$ are exponentially integrable.
We also introduce the following sequences of measures
\[ \tilde{\mu}_n = \sum_{k=1}^n \gamma_k \mu_k + \left( \sum_{k=n+1}^\iy \gamma_k\right) \delta_0 ,\]
\[ \tilde{\nu}_n = \sum_{k=1}^n \gamma_k \nu_k + \left( \sum_{k=n+1}^\iy \gamma_k\right) \delta_0 .\]
One checks using Lebesgue's dominated convergence theorem that the sequences $(\tilde{\mu}_n)$ and $(\tilde{\nu}_n)$ converge respectively toward $\mu$ and $\nu$ for the topology $\tau$. Note also that this sequences are increasing with respect to stochastic domination, so that $\tilde{\nu}_n \stochprec \nu$. For fixed $k$, $\mu_k$ and $\nu_k$ satisfy the hypotheses of Proposition \ref{probabilisticlemma} and thus the same holds for $\tilde{\mu}_n$ and $\tilde{\nu}_n$. Therefore $\tilde{\mu}_n \in D(\tilde{\nu}_n) \subset D(\nu)$. This proves that $\mu \in \overline{D(\nu)}$.

We now prove by contradiction that $\mu \notin D(\nu)$. Assume that $\mu \in D(\nu)$, i.e.  $\mu^{*k} \stochprec \nu^{*k}$ for some $k \geq 1$. Let $s_k=ka_k$ and $t_k=kb_k$.
Fix a sequence $i_1,\dots,i_k$ of nonzero integers. Set $m = \mu_{i_1} * \dots * \mu_{i_k}$ or $m = \nu_{i_1} * \dots * \nu_{i_k}$. We know that $\supp (m) \subset (a,b)$, with $a=\sum_{j=1}^k a_{i_j}$ and $b=\sum_{j=1}^k b_{i_j}$.
It is possible to locate precisely $\supp(m)$ using the inequalities \eqref{ineq1} and \eqref{ineq2}.
\begin{enumerate}
 \item[(a)] If $i_j>k$ for some $j$, then $a \geq a_{k+1} > t_k$ and therefore $\supp(m) \subset (t_k,+\iy)$.
 \item[(b)] If $i_j=k$ for all $j$, then $a = s_k$ and $b = t_k$ and therefore $\supp(m) \subset (s_k,t_k)$.
 \item[(c)] If $i_j \leq k$ for all $j$ and $i_{j_0} < k$ for some $j_0$, then $b \leq b_{k-1}+(k-1)b_k < s_k$
 and therefore $\supp(m) \subset [0,s_k)$.
\end{enumerate}
Consequently,
\[ \mu^{*k}[t_k,+\iy) = \sum_{i_1,\dots,i_k} \gamma_{i_1}\dots\gamma_{i_k} \mu_{i_1}*\dots*\mu_{i_k}[t_k,+\iy) = \sum_{i_1,\dots,i_k \textnormal{ satisfying (a)}} \gamma_{i_1} \dots \gamma_{i_k} = \nu^{*k}[t_k,+\iy) .\]
Moreover, because of (b) and (c), we get that for $s_k\leq t \leq t_k$,
\[ \mu^{*k}[t,t_k) = \gamma_k^k \mu_k^{*k}[t,t_k) = \gamma_k^k \mu_k^{*k}[t,+\iy) .\]
and similarly
\[ \nu^{*k}[t,t_k) = \gamma_k^k \nu_k^{*k}[t, +\iy). \]

We assumed that $\mu^{*k} \stochprec \nu^{*k}$, i.e. $\mu^{*k}[t,+\iy) \leq \nu^{*k}[t,+\iy)$ for all $t$. If $t \leq t_k$, since $\mu^{*k}(t_k,+\iy) = \nu^{*k}(t_k,+\iy)$, we get that $\mu^{*k}[t,t_k) \leq \nu^{*k}[t,t_k)$. Since $\gamma_k>0$, this implies that for all $t \geq s_k$, $\mu_k^{*k}[t,+\iy) \leq \nu_k^{*k}[t,+\iy)$. This contradicts the fact that $\mu_k^{*k} \not \stochprec \nu_k^{*k}$. Therefore $\mu \in \overline{D(\nu)} \setminus D(\nu)$, and so $D(\nu)$ is not closed.
\end{proof}

We now give an example of what can happen if we consider measures with poor integrability properties.

\begin{pr} \label{pr-notconvex}
There exists a probability measure $\nu \in \Prob(\R)$ such that the set 
\begin{equation} \label{notPE} \{ \mu \in \Prob(\R) \st \mu \Mstochprec \nu \} \end{equation}
is not convex.
\end{pr}

The difference between equation \eqref{notPE} and our definition of $D(\nu)$ is that here we do not suppose the measures to be exponentially integrable.

\begin{proof}
We rely on the following fact which we already alluded to (see \cite{feller}, p. 479): there exist two distinct real characteristic functions $\phi_1$ and $\phi_2$ such that $\phi_1^2 = \phi_2^2$ identically.
Consider now the measures $\mu$ and $\nu$ with respective characteristic functions $\phi_1$ and  $\phi_2$, i.e. $\phi_1(t) = \int e^{it} d\mu(t)$ and $\phi_2(t) = \int e^{it} d\nu(t)$. Obviously, we have $\nu \Mstochprec \nu$ and $\mu \Mstochprec \nu$ since $\mu^{*2} = \nu^{*2}$. Let $\chi = \frac 1 2 \mu + \frac 1 2 \nu$ and let us show that $\chi  \not\Mstochprec \nu$. We have
\[\chi^{*2n} = \frac{1}{2^{2n}} \sum_{i=0}^{2n}{2n \choose i} \mu^{*i} *  \nu^{*2n-i} =  \]
\[ = \frac{1}{2^{2n}} \left[ \sum_{i \textnormal{ even}}{{2n \choose i} \nu^{*2n}} + \sum_{i \textnormal{ odd}}{{2n \choose i} \nu^{*2n-1} *\mu} \right]. \]
Thus $\chi^{*2n} \stochprec \nu^{*2n}$, is equivalent to $\nu^{*2n-1} *\mu \stochprec \nu^{*2n}$. Let us show that this is impossible. Indeed, the measures $\nu^{*2n-1} *\mu$ and $\nu^{*2n}$ have real characteristic functions and thus they are symmetric probability measures. Note however that two symmetric probability distributions cannot be compared with $\stochprec$ unless they are equal. But it cannot be that $\nu^{*2n-1} *\mu = \nu^{*2n}$ because their characteristic functions are different ($\phi_1(\xi) = \phi_2(\xi)$ iff. $\phi_1(\xi) = 0$). A similar argument holds for  $\chi^{*2n+1} \notstochprec \nu^{*2n+1}$.
\end{proof}

We conclude this section with few remarks on a relation which is very similar to $\Mstochprec$. It is the analogue of catalytic majorization in quantum information theory (see Section \ref{sec:quantum}).

\begin{defn}
Let $\mu,\nu \in \PE$. We say that $\mu$ is catalytically stochastically dominated by $\nu$ and write $\mu \Cstochprec \nu$ if there exists a probability measure $\pi \in \PE$ such that $\mu*\pi \stochprec \nu*\pi$. 
\end{defn}

The following lemma shows a connection between the two relations.

\begin{lemma} \label{catalysis-measures}
Let $\mu,\nu \in \PE$. Assume $\mu \Mstochprec \nu$. Then $\mu \Cstochprec \nu$. 
\end{lemma}

\begin{proof}
Assume that $\mu^{*n} \stochprec \nu^{*n}$ for some $n$. Let $\pi$ the probability measure defined by
\[ \pi = \frac{1}{n}  \sum_{k=0}^{n-1} \mu^{*k}*\nu^{*(n-1-k)} . \]
Let also $\rho$ be the measure defined by 
\[ \rho = \frac{1}{n} \sum_{k=1}^{n-1} \mu^{*k}*\nu^{*(n-k)} , \]
then one has $\mu * \pi = \frac{1}{n}\mu^{*n} + \rho$ and $\nu * \pi = \frac{1}{n} \nu^{*n}+\rho$, and since $\mu^{*n} \stochprec \nu^{*n}$ this implies $\mu * \pi \stochprec	\nu*\pi$. Since $\pi \in \PE$, we get $\mu \Cstochprec \nu$.
\end{proof}

From Theorem \ref{theo-measures} and Lemma \ref{catalysis-measures} one can easily derive the

\begin{cor}
The analogue of Theorem \ref{theo-measures} is true if we substitute $\Mstochprec$ with $\Cstochprec$.
\end{cor}

\section{Catalytic majorization} \label{sec:quantum}

This section is dedicated to the study of the majorization relation, the notion which was the initial motivation of this work. The majorization relation provides, much as the stochastic domination for probability measures, a partial order on the set of probability vectors. Originally introduced in linear algebra \cite{mo, bhatia}, it has found many application in quantum information theory with the work of Nielsen \cite{nielsen_art, nielsen_book}. We shall not focus on quantum-theoretical aspects of majorization; we refer the interested reader to \cite{an} and references therein. Here, we study majorization by adapting previously obtained results for stochastic domination.


The majorization relation is defined for \emph{probability vectors}, i.e. vectors $x \in \R^\N$ with non-negative components ($x_i \geq 0$) which sum up to one ($\sum_i x_i = 1$). Before defining precisely majorization, let us introduce some notation. For $d \in \N^*$, let $P_d$ be the set of $d$-dimensional probability vectors : $P_d = \{ x \in \R^d \st x_i \geq 0, \sum x_i=1 \}$. Consider also the set of finitely supported probability vectors $P_{<\iy} = \bigcup_{d>0} P_d$. We equip $P_{<\iy}$ with the $\ell_1$ norm defined by $\|x\|_1 = \sum_i |x_i|$. For a vector $x \in P_{< \iy}$, we write $x_{\max}$ for the largest component of $x$ and $x_{\min}$ for its smallest non-zero component. In this section we shall consider only finitely supported vectors. For the general case, see Section \ref{sec:inf_cat}. We shall identify an element $x \in P_d$ with the corresponding element in $P_{d'}$ ($d' >d$) or $P_{<\iy}$ obtained by appending null components at the end of $x$. 

Next, we define $x^{\downarrow}$, the decreasing rearrangement of a vector $x \in P_d$ as the vector which has the same coordinates as $x$ up to permutation and such that $x^\downarrow_i \geq  x^\downarrow_{i+1}$ for all $1 \leq i <d$. We can now define majorization in terms of the ordered vectors:

\begin{defn}
For $x,y \in P_d$ we say that $x$ is majorized by $y$ and we write $x \prec y$ if for all $k \in \{1, \ldots, d\}$
\begin{equation} \label{defmaj}
\sum_{i=1}^{k}{x^\downarrow_i} \leq \sum_{i=1}^{k}{y^\downarrow_i}. \end{equation}
\end{defn}

Note however that there are several equivalent definitions of majorization which do not use the ordering of the vectors $x$ and $y$ (see \cite{bhatia} for further details):

\begin{pr}
The following assertions are equivalent:
\begin{enumerate}
\item $x \prec y$,
\item $\forall t \in \R,  \sum_{i=1}^{d}{|x_i-t|} \leq \sum_{i=1}^{d}{|y_i-t|}$,
\item $\forall t \in \R,  \sum_{i=1}^{d}{(x_i-t)^+} \leq \sum_{i=1}^{d}{(y_i-t)^+}$, where $z^+ = \max(z, 0)$,
\item There is a bistochastic matrix $B$ such that $x= By$.
\end{enumerate}
\end{pr}

There are two operations on probability vectors which are of particular interest to us: the tensor product and the direct sum. For $x=(x_1,\dots,x_d) \in P_d$ and $x'=(x'_1,\dots,x'_{d'}) \in P_{d'}$, we define the tensor product $x \otimes x'$ as the vector $(x_i x'_{j})_{ij} \in P_{dd'}$. We also define the direct sum $x \oplus x'$ as the concatenated vector $(x_1,\dots,x_d,x'_1,\dots,x'_{d'})\in \R^{d+d'}$. Note that if we take $\oplus$-convex combinations, we get probability vectors: $\lambda x \oplus (1-\lambda) x' \in P_{d+d'}$. 

The construction which permits us to use tools from stochastic domination in the framework of majorization is the following (inspired by \cite{kuperberg}): to a probability vector $z \in P_{< \iy}$ we associate a probability measure $\mu_z$ defined by:
\[ \mu_z = \sum z_i \delta_{\log z_i}. \]
These measures behave well with respect to tensor products: 
\[\mu_{x \otimes y} = \mu_x * \mu_y.\]
The connection between majorization and stochastic domination is provided by the following lemma:
\begin{lemma}
\label{lemma-vector-measure}
Let $x,y \in P_{<\iy}$. Assume that $\mu_x \stochprec \mu_y$. Then $x \prec y$.
\end{lemma}

\begin{proof}
We can assume that $x = x^\downarrow$ and $y = y^\downarrow$. Note that
\[\mu_x[t, \iy) = \sum_{i:\log x_i \geq t}{x_i} = \sum_{i:x_i \geq \exp(t)}{x_i}.\]
Thus, for all $u > 0$, $\sum_{i:x_i \geq u}{x_i} \leq \sum_{i:y_i \geq u}{y_i}$. To start, use $u=y_1$ to conclude that $x_1 \leq y_1$. Notice that it suffices to show that $\sum_{i=1}^{k}{x_i} \leq \sum_{i=1}^{k}{y_i}$ only for those $k$ such that $x_k > y_k$ (indeed, if $x_k \leq y_k$, the $(k+1)$-th inequality in \eqref{defmaj} can be deduced from the $k$-th inequality). Consider such a $k$ and let $x_k > u >y_k$. We get:
\[ \sum_{i=1}^{k} {x_i} \leq \sum_{i:x_i \geq u} x_i \leq \sum_{i:y_i \geq u} y_i \leq \sum_{i=1}^{k} {y_i},\]
which completes the proof of the lemma.
\end{proof}

\begin{rk}
The converse of this lemma does not hold. Indeed, consider $x=(0.5, 0.5)$ and $y=(0.9, 0.1)$. Obviously, $x \prec y$ but $1 = \mu_x[\log 0.5, \iy) > \mu_y[\log 0.5, \iy) = 0.9$ and thus $\mu_x \notstochprec \mu_y$.
\end{rk}

We can describe the majorization relation by the sets:
\[S_d(y) = \{ x \in P_d  \text{ s.t. } x \prec y \}, \]
where $y$ is a finitely supported probability vector. Mathematically, such a set is characterized by the following lemma, which is a simple consequence of Birkhoff's theorem on  bistochastic matrices:
\begin{lemma}
For $y$ a $d$-dimensional probability vector, the set $S(y)$ is a polytope whose extreme points are $y$ and its permutations. 
\end{lemma}

The initial motivation for our work was the following phenomena discovered in quantum information theory (see \cite{jp} and respectively \cite{bandy}). It turns out that additional vectors can act as \emph{catalysts} for the majorization relation: there are vectors $x,y,z \in P_{<\iy}$ such that $x \nprec y$ but $x \otimes z \prec y \otimes z$; in such a situation we say that $x$ is catalytically majorized (or \emph{trumped}) by $y$ and we write $x \prec_T y$. Another form of catalysis is provided by \emph{multiple copies} of vectors: we can find vectors $x$ and $y$ such that $x \nprec y$ but still, for some $n \geq 2$, $x^{\otimes n} \prec y^{\otimes n}$; in this case we write $x \prec_M y$. We have thus two new order relations on probability vectors, analogues of $\Cstochprec$ and respectively $\Mstochprec$. As before, for $y \in P_d$, we introduce the sets
\[ T_d(y) = \{ x \in P_d \st x \prec_T y \}, \]
and
\[ M_d(y) = \{ x \in P_d \st x \prec_M y \}. \]

It turns out that the relations $\prec_T$ and $\prec_M$ (and thus the sets $T_d(y)$ and $M_d(y)$) are not as simple as $\prec$ and $S_d(y)$. It is known that the inclusion $M_d(y) \subset T_d(y)$ holds (this is the analogue of Lemma \ref{catalysis-measures}) and that it can be strict \cite{fdy}. In general, the sets $T_d(y)$ and $M_d(y)$ are neither closed nor open, and although $T_d(y)$ is known to be convex, nothing is known about the convexity of $M_d(y)$ (such questions have been intensively studied in the physical literature; see \cite{dk, djfy} and the references therein). As explained in \cite{an} it is natural from a mathematical point of view to introduce the sets $T_{<\iy}(y) = \bigcup_{d \in \N} T_d(y)$ and $M_{<\iy}(y) = \bigcup_{d \in \N} M_d(y)$.
A key notion in characterizing them is \emph{Schur-convexity}:

\begin{defn}
A function $f : P_d \to \R$ is said to be 
\begin{itemize}
\item Schur-convex if $f(x) \leq f(y)$ whenever $x \prec y$,
\item Schur-concave if $f(x) \geq f(y)$ whenever $x \prec y$,
\item strictly Schur-convex if $f(x) < f(y)$ whenever $x \precneqq y$,
\item strictly Schur-concave if $f(x) > f(y)$ whenever $x \precneqq y$,
\end{itemize}
where $x \precneqq y$ means $x \prec y$ and $x^\downarrow \neq y^\downarrow$.
\end{defn}

Examples are provided as follows: if $\Phi:\R \to \R$ is a (strictly) convex/concave function, then the following function $h:P_d \to \R$ defined by $h(x_1,\dots,x_d) = \Phi(x_1) + \cdots + \Phi(x_d)$ is (strictly) Schur-convex/Schur-concave.

For $x \in P_d$ and $p \in \R$, we define $N_p(x)$ as 
\[ N_p(x) = \sum_{\begin{subarray}{c}1 \leq i \leq d \\ x_i > 0 \end{subarray}} x_i^p .\]
We will also use the Shannon entropy $H$
\[ H(x) = -\sum_{i=1}^d x_i \log x_i .\]
Note that $-H(x)$ is the derivative of $p \mapsto N_p(x)$ at $p=1$ and that $N_0(x)$ is the number of non-zero components of the vector $x$.
These functions satisfy the following properties:
\begin{enumerate}
\item If $p > 1$, $N_p$ is strictly Schur-convex on $P_{<\iy}$.
\item If $0 < p < 1$, $N_p$ is strictly Schur-concave on $P_{<\iy}$.
\item If $p < 0$, $N_p$ is strictly Schur-convex on $P_d$ for any $d$. However, for $p<0$, it is not possible to compare vectors with a different number of non-zero components.
\item $H$ is strictly Schur-concave on $P_{<\iy}$.
\end{enumerate}

One possible way of describing the relations $\prec_M$ and $\prec_T$ is to find a family (the smallest possible) of Schur-convex functions which characterizes them. In this direction, Nielsen conjectured the following result:

\begin{conj}
\label{conj-nielsen}
Fix a vector $y \in P_d,$ with nonzero coordinates. Then $\overline{T_d(y)}=\overline{M_d(y)}$ and they both are equal to the set of $x \in P_d$ satisfying
\begin{enumerate}
\item[(C1)] For $p \geq 1$, $N_p(x) \leq N_p(y)$.
\item[(C2)] For $0 < p \leq 1$, $N_p(x) \geq N_p(y)$.
\item[(C3)] For $p < 0$, $N_p(x) \leq N_p(y)$.
\end{enumerate}
\end{conj}
Here, the closures are taken in $\R^d$ (recall that neither $M_d(y)$ nor $T_d(y)$ is closed). By the previous remarks, any vector in  $T_d(y)$ or $M_d(y)$ (and by continuity, also in the closures) must satisfy conditions (C1-C3). Recently, Turgut \cite{t1, t2} provided a complete characterization of the set $T_d(y)$, which implies in particular that Nielsen's conjecture is true for $\overline{T_d(y)}$. His method, completely different from ours, consists in solving a discrete approximation of the problem using elementary algebraic techniques. Note however that the inclusion $M_d(y) \subset T_d(y)$ is strict in general, and thus the characterization of $\overline{M_d(y)}$ is still open. We shall now focus on the set $M_d(y)$. Conjecture \ref{conj-nielsen} can be reformulated as follows: if $x,y \in P_d$ and satisfy (C1-C3), then there exists a sequence $(x_n)$ in $M_d(y)$ such that $(x_n)$ converges to $x$. If we relax the condition that $x_n$ and $y$ have the same dimension, we can prove the following two theorems:

\begin{theo}
\label{th1}
If $x,y \in P_d$ and satisfy (C1), then there exists a sequence $(x_n)$ in $M_{< \iy}(y)$ such that $(x_n)$ converges to $x$ in $\ell_1$-norm.
\end{theo}
\begin{theo}
\label{th2}
If $x,y \in P_d$ and satisfy (C1-C2), then there exists a sequence $(x_n)$ in $M_{d+1}(y)$ such that $(x_n)$ converges
 to $x$.
\end{theo}

Since $M_d(y) \subset T_d(y)$, both theorems have direct analogues for $T_{< \iy}(y)$ and respectively $T_{d+1}(y)$. Theorem \ref{th1} restates the authors' previous result in \cite{an}; however, the proof presented in the next section is more transparent than the previous one. Theorem \ref{th2} answers a question of \cite{an}. It is an intermediate result between Theorem \ref{th1} and Conjecture \ref{conj-nielsen}.

\section{Proof of the theorems}\label{sec:proofs}

We show here how to derive Theorems \ref{th1} and \ref{th2}. We first state a proposition which is the translation of Proposition \ref{probabilisticlemma} in terms of majorization.

\begin{pr}
\label{lemma-main}
Let $x,y \in P_{<\iy}$. Assume that $x$ and $y$ have nonzero coordinates, and respective dimensions $d_x$ and $d_y$.
Assume that
\begin{enumerate}
\item $x_{\min} < y_{\min}$.
\item $x_{\max} < y_{\max}$.
\item $H(x) > H(y)$.
\item $N_p(x) < N_p(y)$ for all $p \in ]1,+\iy[$.
\item $N_p(x) > N_p(y)$ for all $p \in ]-\iy,1[$.
\end{enumerate}
Then there exists an integer $N$ such that for all $n \geq N$, we have $x^{\otimes n} \prec y^{\otimes n}$.
\end{pr}

It is important to notice that since $N_0(x)=d_x$ and $N_0(y)=d_y$, the conditions of the proposition can be satisfied only when $d_x>d_y$. This is the main reason why our approach fails to prove Conjecture \ref{conj-nielsen}.

\begin{proof} One checks that the probability measures $\mu_x$ and $\mu_y$ associated to the vectors $x$ and $y$ satisfy the hypotheses of Proposition \ref{probabilisticlemma}. Indeed, for $p \in \R$, one has
\[N_p(x) = \int \el d\mu_x, \ \ \textnormal{ with } \lambda = p-1.\]
As $\mu_x^{*n} = \mu_{x^{\otimes n}}$, there exists a integer $N$ such that for $n \geq N$, we have $\mu_{x^{\otimes n}} \stochprec \mu_{y^{\otimes n}}$. It remains to apply the Lemma \ref{lemma-vector-measure} in order to complete the proof.
\end{proof}
The main idea used in the following proofs is to slightly modify the vector $x$ so that the couple ($x$, $y$) satisfies the hypotheses of Proposition \ref{lemma-main}.

\begin{proof}[Proof of Theorem \ref{th1}]
Let $x,y \in P_d$ satisfying $N_p(x) \leq N_p(y)$ for all $p \geq 1$. Since $N_1(x)=N_1(y)=1$ and $-H = \frac{dN_p}{dp}|_{p=1}$, we also have $-H(x) \leq -H(y)$.
For $0 < \e < \frac{d}{d+1} x_{\min}$, define $x_\e \in P_{d+1}$ by
\[ x_\e = (x_1-\frac{\e}{d},\dots,x_d-\frac{\e}{d},\e) .\]
One checks that $x_\e \precneqq x$ and therefore $N_p(x_\e) < N_p(x) \leq N_p(y)$ for any $p >1$, and $-H(x_\e) < -H(x) \leq -H(y)$. Since $-H = \frac{dN_p}{dp}|_{p=1}$ and the function $p \mapsto N_p(\cdot)$ is continuous, this means that there exists some $0 < p_\e<1$ such that $N_p(x_\e) \geq N_p(y)$ for any $p \in [p_\e,1]$. 
Choose an integer $k \geq 2$, depending on $\e$, such that 
\[k > \max\{d^{1/(1-p_\e)}\e^{-p_\e/(1-p_\e)}, \frac{\e}{y_{\min}}, d\}\]
and define $x_{\e,k} \in P_{<\iy}$ as
\[ x_{\e,k} = (x_1-\frac{\e}{d},\dots,x_d-\frac{\e}{d},\underbrace{\frac{\e}{k},\dots,\frac{\e}{k}}_{k \textnormal{ times}}). \]
For any $0 \leq p\leq p_\e$ we have
\[ N_p(x_{\e,k}) \geq k \left(\frac{\e}{k}\right)^p > d \geq N_p(y), \]
and for any $p<0$ we have
\[ N_p(x_{\e,k}) \geq k \left(\frac{\e}{k}\right)^p > d y_{\min}^p \geq N_p(y). \]
We also have $x_{\e,k} \precneqq x_\e$ and therefore $N_p(x_{\e,k}) > N_p(x_\e) \geq N_p(y)$ for $p_\e \leq p<1$. Similarly, $N_p(x_{\e,k}) < N_p(x_\e) \leq N_p(y)$ for $p>1$. This means that $x_{\e,k}$ and $y$ satisfy the hypotheses of Proposition \ref{lemma-main}, and therefore $x_{\e,k} \in M_{< \iy}(y)$. Since $||x_{\e,k}-x||_1 \leq 2\e$ and $\e$ can be chosen arbitrarily small, this completes the proof of the theorem.
\end{proof}

\begin{proof}[Proof of Theorem \ref{th2}]
Let $x,y \in P_d$ satisfying $N_p(x) \leq N_p(y)$ for $p \geq 1$ and $N_p(x) \geq N_p(y)$ for $0 \leq p \leq 1$. As in the previous proof, we consider for $0 < \e < \frac{d}{d+1} x_{\min}$ the vector $x_\e$ defined as
\[ x_\e = (x_1-\frac{\e}{d},\dots,x_d-\frac{\e}{d},\e) .\]
We are going to show using Proposition \ref{lemma-main} that for $\e$ small enough, $x_\e$ is in $M_{d+1}(y)$. Note that $x_\e \precneqq x$, and therefore $N_p(x_\e) < N_p(x) \leq N_p(y)$ for $p > 1$, and $N_p(x_\e) > N_p(x) \geq N_p(y)$ for $0 < p <1$. Also, since $N_0(x_\e)=d+1$ and $N_0(y)=d$, there exists by continuity a number $p_0<0$ (not depending on $\e$) such that $N_p(y) < d+1$ for all $p \in [p_0,0]$. Thus for $p \in [p_0,0]$ we have 
\[ N_p(x_\e) \geq N_0(x_\e) = d+1 > N_p(y) .\]
It remains to notice that for $\e < d^{1/p_0}y_{\min}$, we have for any $p\leq p_0$
\[ N_p(x_\e) \geq \e^p > d y_{\min}^p \geq N_p(y) .\]
We checked that $x_\e$ and $y$ satisfy the hypotheses of Proposition \ref{lemma-main}, and therefore $x_\e \in M_{d+1}(y)$. Since  $||x_\e-y||_1 \leq 2\e$ and $\e$ can be chosen arbitrarily small, this completes the proof of the theorem.
\end{proof}

\section{Infinite dimensional catalysis}\label{sec:inf_cat}

In light of the recent paper \cite{owari}, we investigate the majorization relation and its generalizations for infinitely-supported probability vectors. Let us start by adapting the key tools used in the previous section to this non-finite setting. 

First, note that when defining the decreasing rearrangement $x^\downarrow$ of a vector $x$, we shall ask that only the \emph{non-zero} components of $x$ and $x^\downarrow$ should be the same up to permutation. The majorization relation $\prec$ extends trivially to $P_\iy$, the set of (possibly infinite) probability vectors. The same holds for the relations $\prec_M$ and $\prec_T$ (note however that for $\prec_T$, we allow now infinite-dimensional catalysts). 

Note that for a general probability vector, there is no reason that $N_p$ for $p \in (0, 1)$ or $H$ should be finite. He have thus to replace the hypothesis (C1) by the following one:

\begin{enumerate}
\item[(C1')] For $p \geq 1$, $N_p(x) \leq N_p(y)$ and $H(x) < \iy$.
\end{enumerate}

Notice however that the inequalities $N_p(x) \leq N_p(y)$ for $p \to 1^+$ imply that $H(y) \leq H(x) < \iy$ and thus both entropies are finite. 
\begin{theo}
\label{th3}
If $x,y \in P_\iy$ and satisfy (C1'), then, for all $\e >0$ there exist finitely supported vectors $x_\e, y_\e \in P_{<\iy}$ and $n \in \N$ such that $\|x-x_\e\|_1 \leq \e$, $\|y-y_\e\|_1 \leq \e$ and $x_\e^{\otimes n} \prec y_\e^{\otimes n}$.
\end{theo}
\begin{proof} Fix $\e >0$ small enough. If $y$ has infinite support, consider the truncated vector $y_\e = (y_1 + R(\e), y_2, \ldots, y_{N(\e)})$, where $N(\e)$ and $R(\e)$ are such that $R(\e) = \sum_{i=N(\e)+1}^{\iy}{y_i} \leq \e$; otherwise put $y_\e = y$. Clearly, we have $\|y - y_\e\|_1 \leq 2\e$ and $N_p(y_\e) \geq N_p(y)$ for all $p>1$.
If the vector $x$ is finite, use Theorem \ref{th1} with $x_\e = x$ and $y_\e$ to conclude. Otherwise, consider  $M(\e)$ such that $S(\e) = \sum_{i=M(\e)+1}^{\iy}{x_i} \leq \e $ and define the vector 
\[ x_\e = (x_1, x_2, \ldots, x_{M(\e)}, \underbrace{\frac{S(\e)}{k}, \frac{S(\e)}{k}, \dots, \frac{S(\e)}{k}}_{k \textnormal{ times}}), \] 
where $k$ is a constant depending on $\e$ which will be chosen later. For all $k \geq 1$, $x_\e$ is a finite vector of size $M(\e) + k$ and we have $\|x-x_\e\|_1 \leq 2\e$. Let us now show that we can chose $k$ such that $N_p(x_\e) \leq N_p(x)$ for all $p \geq 1$. In order to do this, consider the function $\phi:(1, \iy) \rightarrow \R_+$ 
\[\phi(p) = \left[ \frac{S(\e)^p}{\sum_{i=M(\e)+1}^{\iy}{x_i^p}} \right]^\frac{1}{p-1}. \]  
The function $\phi$ takes finite values on $(1, \iy)$ and $\lim_{p \to \iy}{\phi(p) = \frac{S(\e)}{x_{M(\e)+1}}} < \iy$. Moreover, as the Shannon entropy of $x$ is finite, one can also show that $\lim_{p \to 1^+}{\phi(p)} < \iy$. Thus, the function $\phi$ is bounded and we can choose $k \in \N$ such that $k \geq \phi(p)$ for all $p \geq 1$. This implies that
\[N_p(x_\e) - N_p(x) = k \left( \frac{S(\e)}{k} \right)^p - \sum_{i=M(\e)+1}^{\iy}{x_i^p} \leq 0.\]
In conclusion, we have found two finitely supported vectors $x_\e$ and $y_\e$ such that $\|x - x_\e\|_1 \leq 2\e$, $\|y - y_\e\|_1 \leq 2\e$ and $N_p(x_\e) \leq N_p(y_\e)$ for all $p \geq 1$. To conclude, it suffices to apply Theorem \ref{th1} to $x_\e$ and $y_\e$.
\end{proof}

\bigskip

Address : \\
Université de Lyon, \\
Université Lyon 1, \\
CNRS, UMR 5208 Institut Camille Jordan, \\
Batiment du Doyen Jean Braconnier, \\
43, boulevard du 11 novembre 1918, \\
F - 69622 Villeurbanne Cedex, \\
France\\
\\
Email: aubrun@math.univ-lyon1.fr, nechita@math.univ-lyon1.fr

\end{document}